\documentclass{article}
\pdfoutput=1
\usepackage[nonatbib,final]{nips_2016}

\usepackage[utf8]{inputenc} 
\usepackage[T1]{fontenc}    
\usepackage{url}            
\usepackage{booktabs}       
\usepackage{amsfonts}       
\usepackage{amsmath}
\usepackage{amsthm}
\usepackage{amssymb}
\usepackage{nicefrac}       
\usepackage{microtype}      
\usepackage{graphicx}
\usepackage{subfig}

\newtheorem{lemma}{Lemma}
\newtheorem{corollary}{Corollary}

\newtheorem{theorem}{Theorem}

\DeclareMathOperator*{\argmax}{arg\,max}

\title{Reduced Space and Faster Convergence in Imperfect-Information \nolinebreak Games \nolinebreak via \nolinebreak Regret-Based \nolinebreak Pruning}

\author{
	Noam Brown \\
	Computer Science Department\\
	Carnegie Mellon University\\
	Pittsburgh, PA 15217 \\
	\texttt{noamb@cmu.edu} \\
	\And
	Tuomas Sandholm \\
	Computer Science Department \\
	Carnegie Mellon University \\
	Pittsburgh, PA 15217 \\
	\texttt{sandholm@cs.cmu.edu}
}

\begin{document}

\maketitle


\begin{abstract}
\emph{Counterfactual Regret Minimization} (CFR) is the most popular iterative algorithm for solving zero-sum imperfect-information games. \emph{Regret-Based Pruning} (RBP) is an improvement that allows poorly-performing actions to be temporarily pruned, thus speeding up CFR. We introduce \emph{Total RBP}, a new form of RBP that reduces the space requirements of CFR as actions are pruned. We prove that in zero-sum games it asymptotically prunes any action that is not part of a best response to some Nash equilibrium. This leads to provably faster convergence and lower space requirements. Experiments show that Total RBP results in an order of magnitude reduction in space, and the reduction factor increases with game size.
\end{abstract}

\section{Introduction}
Imperfect-information extensive-form games model strategic multi-step scenarios between agents with hidden information, such as auctions, security interactions (both physical and virtual), negotiations, and military situations. Typically in imperfect-information games, one wishes to find a Nash equilibrium, which is a profile of strategies in which no player can improve her outcome by unilaterally changing her strategy. A linear program can find an exact Nash equilibrium in two-player zero-sum games containing fewer than about $10^8$ nodes~\cite{Gilpin07:Lossless_long}. For larger games, iterative algorithms are used to converge to a Nash equilibrium. There are a number of such iterative algorithms~\cite{Nesterov05:Excessive,Hoda10:Smoothing,Pays14:Interior,Kroer15:Faster}, the most popular of which is \emph{Counterfactual Regret Minimization} (CFR)~\cite{Zinkevich07:Regret}. CFR minimizes regret independently at each decision point in the game. CFR+, a variant of CFR, was used to essentially solve Limit Texas Hold'em, the largest imperfect-information game ever to be essentially solved~\cite{Bowling15:Heads-up}.

Both computation time and storage space are difficult challenges when solving large imperfect-information games. For example, solving Limit Texas Hold'em required nearly 8 million core hours and a complex, domain-specific streaming compression algorithm to store the 262 TiB of uncompressed data in only 10.9 TiB. This data had to be repeatedly decompressed from disk into memory and then compressed back to disk in order to run CFR+~\cite{Tammelin15:Solving}.

\emph{Regret Based Pruning} (RBP) is an improvement to CFR that greatly reduces the computation time needed to solve large games by temporarily pruning suboptimal actions~\cite{Brown15:Regret-Based}. Specifically, if an action has negative regret, then RBP skips that action for the minimum number of iterations it would take for its regret to become positive in CFR. The skipped iterations are then ``made up'' in a single iteration once pruning ends.

In this paper we introduce a new form of RBP which we coin \emph{Total RBP}. It alters the starting and ending conditions for pruning, and changes the way regrets are updated once pruning ends. We refer to the prior form of RBP as \emph{Interval RBP} to differentiate it from our new method.

A primary advantage of Total RBP is that in addition to faster convergence, it also reduces the \emph{space} requirements of CFR over time. Specifically, once pruning begins on a branch, Total RBP ensures the regrets currently stored on that branch will never be needed again. In other words, all the data stored for a branch that is pruned can be discarded, and the space allocated to that branch can be freed. Space need not be reallocated for that branch until pruning ends and the branch cannot immediately be pruned again. In Section~\ref{sec:rbpspace}, we prove that after enough iterations are completed, space for certain pruned branches will \emph{never} need to be allocated again. Specifically, we prove that Total RBP need only asymptotically store actions that have positive probability in a best response to a Nash equilibrium. This is extremely advantageous when solving large imperfect-information games, which are often constrained by space and in which the set of best response actions may be orders of magnitude smaller than the size of the game.

While Total RBP still requires enough memory to store the entire game in the early iterations, recent work has shown that these early iterations can be skipped by first solving a low-memory abstraction of the game and then using its solution to warm starting CFR in the full game~\cite{Brown16:Strategy-Based}. Total RBP's reduction in space is also helpful to the \emph{Simultaneous Abstraction and Equilibrium Finding} (SAEF) algorithm~\cite{Brown15:Simultaneous}, which starts CFR with a small abstraction of the game and progressively expands the abstraction while also solving the game. SAEF's space requirements increase the longer the algorithm runs, and may eventually exceed the constraints of a system. Total RBP can counter this increase in space by eliminating the need to store suboptimal paths of the game tree.

While prior work on Interval RBP has shown empirical evidence of better performance, this paper proves that CFR converges faster when using Total RBP, because certain suboptimal actions will only need to be traversed $O\big(\ln(T)\big)$ times over $T$ iterations.


The magnitude of these gains in speed and space varies depending on the game. It is possible to construct games where Total RBP provides no benefit. However, if there are many suboptimal actions in the game---as is frequently the case in large games---Total RBP can speed up CFR by multiple orders of magnitude and require orders of magnitude less space. Our experiments show an order of magnitude space reduction already in medium-sized games, and a reduction factor increase with game size.

\vspace{-0.075in}
\section{Background}
\vspace{-0.075in}
\label{sec:background}
This section presents the notation used in the rest of the paper. An imperfect-information extensive-form game
has a finite set of players, $\mathcal{P}$. $H$ is the set of all possible histories (nodes) in the game tree, represented as a sequence of actions, and includes the empty history. $A(h)$ is the actions available in a history and $P(h) \in \mathcal{P} \cup c$ is the player who acts at that history, where $c$ denotes chance. Chance plays an action $a \in A(h)$ with a fixed probability $\sigma_c(h,a)$ that is known to all players. The history $h'$ reached after action $a$ in $h$ is a child of $h$, represented by $h \cdot a = h'$, while $h$ is the parent of $h'$. More generally, $h'$ is an ancestor of $h$ (and $h$ is a descendant of $h'$), represented by $h' \sqsubset h$, if there exists a sequence of actions from $h'$ to $h$. $Z \subseteq H$ are terminal histories for which no actions are available. For each player $i \in \mathcal{P}$, there is a payoff function $u_i: Z\rightarrow \Re$. If $P = \{1, 2\}$ and $u_1 = -u_2$, the game is two-player zero-sum. Define $\Delta_i = \max_{z \in Z} u_i(z) - \min_{z \in Z} u_i(z)$ and $\Delta = \max_i\Delta_i$.

Imperfect information is represented by \emph{information sets} for each player $i \in \mathcal{P}$ by a partition $\mathcal{I}_i$ of $h \in H : P(h) = i$. For any information set $I \in \mathcal{I}_i$, all histories $h, h' \in I$ are indistinguishable to player $i$, so $A(h) = A(h')$.
$I(h)$ is the information set $I$ where $h \in I$. $P(I)$ is the player $i$ such that $I \in \mathcal{I}_i$. $A(I)$ is the set of actions such that for all $h \in I$, $A(I) = A(h)$.
$|A_i| = \max_{I \in \mathcal{I}_i}|A(I)|$ and $|A| = \max_i|A_i|$. Define $U(I)$ to be the maximum payoff reachable from a history in $I$, and $L(I)$ to be the minimum. That is, $U(I) = \max_{z \in Z, h \in I : h \sqsubseteq z} u_{P(I)}(z)$ and $L(I) = \min_{z \in Z, h \in I : h \sqsubseteq z} u_{P(I)}(z)$. Define $\Delta(I) = U(I) - L(I)$ to be the range of payoffs reachable from a history in $I$. Similarly $U(I,a)$, $L(I,a)$, and $\Delta(I,a)$ are the maximum, minimum, and range of payoffs (respectively) reachable from a history in $I$ after taking action $a$.
Define $D(I,a)$ to be the set of information sets reachable by player $P(I)$ after taking action $a$. Formally, $I' \in D(I,a)$ if for some history $h \in I$ and $h' \in I'$, $h \cdot a \sqsubseteq h'$ and $P(I) = P(I')$.

A strategy $\sigma_i(I)$ is a probability vector over $A(I)$ for player $i$ in information set $I$. The probability of a particular action $a$ is denoted by $\sigma_i(I,a)$. Since all histories in an information set belonging to player $i$ are indistinguishable, the strategies in each of them must be identical. That is, for all $h \in I$, $\sigma_i(h) = \sigma_i(I)$ and $\sigma_i(h,a) = \sigma_i(I,a)$. Define $\sigma_i$ to be a probability vector for player $i$ over all available strategies $\Sigma_i$ in the game. A strategy profile $\sigma$ is a tuple of strategies, one for each player. $u_i(\sigma_i, \sigma_{-i})$ is the expected payoff for player $i$ if all players play according to the strategy profile $\langle \sigma_i, \sigma_{-i} \rangle$. If a series of strategies are played over $T$ iterations, then $\bar{\sigma}^T_i = \frac{\sum_{t \in T} \sigma^t_i}{T}$.

$\pi^{\sigma}(h) = \Pi_{h' \rightarrow a \sqsubseteq h} \sigma_{P(h)}(h,a)$ is the joint probability of reaching $h$ if all players play according to $\sigma$. $\pi^{\sigma}_i(h)$ is the contribution of player $i$ to this probability (that is, the probability of reaching $h$ if all players other than $i$, and chance, always chose actions leading to $h$). $\pi^{\sigma}_{-i}(h)$ is the contribution of all players other than $i$, and chance.
$\pi^{\sigma}(h,h')$ is the probability of reaching $h'$ given that $h$ has been reached, and $0$ if $h \not \sqsubset h'$. In a \emph{perfect-recall} game, $\forall h, h' \in I \in \mathcal{I}_i$, $\pi_i(h) = \pi_i(h')$. In this paper we focus on perfect-recall games.
Therefore, for $i = P(I)$ define $\pi_i(I) = \pi_i(h)$ for $h \in I$. Moreover, $I' \sqsubset I$ if for some $h' \in I'$ and some $h \in I$, $h' \sqsubset h$. Similarly, $I' \cdot a \sqsubset I$ if $h' \cdot a \sqsubset h$. The average strategy $\bar{\sigma}_i^T(I)$ for an information set $I$ is defined as $\bar{\sigma}_i^T(I) = \frac{\sum_{t \in T} \pi_i^{\sigma^t_i}(I)\sigma_i^t(I)}{\sum_{t \in T} \pi_i^{\sigma^t}(I)}$.

A \emph{best response} to $\sigma_{-i}$ is a strategy $\sigma^*_i$ such that $u_i(\sigma^*_i, \sigma_{-i}) = \max_{\sigma'_i \in \Sigma_i} u_i(\sigma'_i, \sigma_{-i})$. A \emph{Nash equilibrium} $\sigma^*$ is a strategy profile where every player plays a best response: $\forall i$, $u_i(\sigma^*_i, \sigma^*_{-i}) = \max_{\sigma'_i \in \Sigma_i} u_i(\sigma'_i, \sigma^*_{-i})$. A \emph{Nash equilibrium strategy} for player $i$ as a strategy $\sigma_i$ that is part of any Nash equilibrium. In two-player zero-sum games, if $\sigma_i$ and $\sigma_{-i}$ are both Nash equilibrium strategies, then $\langle \sigma_i, \sigma_{-i} \rangle$ is a Nash equilibrium. An \emph{$\epsilon$-equilibrium} as a strategy profile $\sigma^*$ such that $\forall i$, $u_i(\sigma^*_i, \sigma^*_{-i}) + \epsilon \ge \max_{\sigma'_i \in \Sigma_i} u_i(\sigma'_i, \sigma^*_{-i})$.

\subsection{Counterfactual Regret Minimization}
\emph{Counterfactual Regret Minimization (CFR)} is a popular algorithm for extensive-form games in which the strategy vector for each information set is determined according to a regret-minimization algorithm~\cite{Zinkevich07:Regret}. We use \emph{regret matching (RM)}~\cite{Hart00:Simple} as the regret-minimization algorithm.

The analysis of CFR makes frequent use of \emph{counterfactual value}. Informally, this is the expected utility of an information set given that player $i$ tries to reach it. For player $i$ at information set $I$ given a strategy profile $\sigma$, this is defined as
\begin{equation}
v^{\sigma}(I) = \sum_{h \in I} \Big(\pi_{-i}^{\sigma}(h)\sum_{z \in Z}\big(\pi^{\sigma}(h,z)u_i(z)\big)\Big)
\label{eq:counterfactual}
\end{equation}
The counterfactual value of an action $a$ is
\begin{equation}
v^{\sigma}(I,a) = \sum_{h \in I} \Big(\pi_{-i}^{\sigma}(h)\sum_{z \in Z}\big(\pi^{\sigma}(h \cdot a,z)u_i(z)\big)\Big)
\label{eq:counterfactuala}
\end{equation}

A \emph{counterfactual best response}~\cite{Moravcik16:Refining} (CBR) is a strategy similar to a best response, except that it maximizes counterfactual value even at information sets that it does not reach due to its earlier actions. Specifically, a counterfactual best response to $\sigma_{-i}$ is a strategy $CBR(\sigma_{-i})$ such that if $CBR(\sigma_{-i})(I,a) > 0$ then $v^{\langle CBR(\sigma_{-i}), \sigma_{-i} \rangle}(I,a) = \max_{a'} v^{\langle CBR(\sigma_{-i}), \sigma_{-i} \rangle}(I,a')$. The \emph{counterfactual best response value} $CBV^{\sigma_{-i}}(I)$ is similar to counterfactual value, except that player $i = P(I)$ plays according to a CBR to $\sigma_{-i}$. Formally, $CBV^{\sigma_{-i}}(I) = v^{\langle CBR_i(\sigma_{-i}), \sigma_{-i} \rangle}(I)$.

Let $\sigma^t$ be the strategy profile used on iteration $t$. The \emph{instantaneous regret}
on iteration $t$ for action $a$ in information set $I$ is
$r^t(I,a) = v^{\sigma^t}(I,a) - v^{\sigma^t}(I)$
and the \emph{regret} for action $a$ in $I$
on iteration $T$ is
$R^T(I,a) = \sum_{t \in T} r^t(I,a)$.
Additionally, $R^T_+(I,a) = \max\{R^T(I,a), 0 \}$ and $R^T(I) = \max_a\{R_+^T(I,a)\}$. Our analysis of Total RBP will occasionally reference the \emph{potential function} of $R(I)$, defined as $\Phi(R^T(I)) = \sum_{a \in A(I)} \big(R_+^T(I,a)\big)^2$. Regret
for player $i$ in the entire game is
$R_i^T = \max_{\sigma_i' \in \Sigma_i} \sum_{t \in T} \big(u_i(\sigma'_i, \sigma_{-i}^t) - u_i(\sigma^t_i, \sigma_{-i}^t)\big)$
.

In regret matching, a player picks a distribution over actions in an information set in proportion to the positive regret on those actions. Formally, on each iteration $T+1$, player $i$ selects actions $a \in A(I)$ according to probabilities
\vspace{-0.04in}
\begin{equation}
\sigma^{T+1}(I,a) =
\begin{cases}
\frac{R^T_+(I,a)}{\sum_{a' \in A(I)}R_+^T(I,a')}, & \text{if}\ \sum_{a'}R^T_+(I,a') > 0 \\
\frac{1}{|A(I)|}, & \text{otherwise}
\end{cases}
\label{eq:rm}
\end{equation}
If a player plays according to RM on every iteration then on iteration $T$, $R^T(I) \le \Delta(I)\sqrt{|A(I)|}\sqrt{T}$.

If a player plays according to CFR in every iteration then
$R_i^T \le \sum_{I \in \mathcal{I}_i} R^T(I)$.
So, as $T \rightarrow \infty$, $\frac{R_i^T}{T} \rightarrow 0$. In two-player zero-sum games, if both players' average regret $\frac{R_i^T}{T} \le \epsilon$, their average strategies $\langle \bar{\sigma}^T_1, \bar{\sigma}^T_2 \rangle$ form a $2\epsilon$-equilibrium~\cite{Waugh09:Abstraction}. Thus, CFR constitutes an anytime algorithm for finding an $\epsilon$-Nash equilibrium in zero-sum games.

\subsection{Partial Pruning and Interval Regret-Based Pruning}
\label{sec:pruning}
This section reviews forms of pruning that allow parts of the game tree to be skipped in CFR. Typically, regret is updated by traversing each node in the game tree separately for each player, and calculating the contribution of a history $h \in I$ to $r^t(I,a)$ for each action $a \in A(I)$. If a history $h$ is reached in which $\pi_{-i}^{\sigma^t}(h) = 0$ (that is, an opponent's reach is zero), then from (\ref{eq:counterfactual}) and (\ref{eq:counterfactuala}) the strategy at $h$ contributes nothing on iteration $t$ to the regret of $I(h)$ (or to the information sets above it). Moreover, any history that would be reached beyond $h$ would also contribute nothing to its information set's regret because $\pi^{\sigma^t}_{-i}(h') = 0$ for every history $h'$ where $h \sqsubset h'$ and $P(h') = P(h)$. Thus, when traversing the game tree for player $i$, there is no need to traverse beyond any history $h$ when $\pi_{-i}^{\sigma^t}(h) = 0$. The benefit of this form of pruning, which we refer to as \emph{partial pruning}, varies depending on the game, but empirical results show a factor of $30$ improvement in some games~\cite{Lanctot09:Monte}.


While partial pruning allows one to prune paths that an \emph{opponent} reaches with zero probability, \emph{interval regret-based pruning} (Interval RBP) allows one to also prune paths that the \emph{traverser} reaches with zero probability~\cite{Brown15:Regret-Based}. However, this pruning is necessarily temporary. Consider an action $a \in A(I)$ such that $\sigma^t(I,a) = 0$, and assume that it is known action~$a$ will not be played with positive probability until some far-future iteration $t'$ (in RM, this would be the case if $R^t(I,a) \ll 0$). Since action~$a$ is played with zero probability on iteration $t$, so from (\ref{eq:counterfactual}) the strategy played and reward received following action~$a$ (that is, in $D(I,a)$) will not contribute to the regret for any information set preceding action~$a$ on iteration $t$. In fact, what happens in $D(I,a)$ has no bearing on the rest of the game tree until iteration $t'$ is reached. So one could, in theory, ``procrastinate'' in deciding what happened beyond action~$a$ on iteration $t$, $t+1$, ..., $t' - 1$ until iteration $t'$.

However, upon reaching iteration $t'$, rather than individually making up the $t' - t$ iterations over $D(I,a)$, one can instead do a \emph{single} iteration, playing against the average of the opponents' strategies in the $t' - t$ iterations that were missed, and declare that strategy was played on all the $t' - t$ iterations. This accomplishes the work of the $t' - t$ iterations in a single traversal. Moreover, since player $i$ never plays action $a$ with positive probability between iterations $t$ and $t'$, that means every \emph{other} player can apply partial pruning on that part of the game tree for iterations $t' - t$, and skip it completely. This, in turn, means that player $i$ has free rein to play whatever they want in $D(I,a)$ without affecting the regrets of the other players. In light of that, and of the fact that player $i$ gets to decide what is played in $D(I,a)$ after knowing what the other players have played, player $i$ might as well play a strategy that ensures zero regret for all information sets $I' \in D(I,a)$ in the iterations $t$ to $t'$. A CBR to the average of the opponent strategies on the $t' - t$ iterations would qualify as such a zero-regret strategy.

Interval regret-based pruning only allows a player to skip traversing $D(I,a)$ for as long as $\sigma^t(I,a) = 0$. Thus, in RM,  if $R^{t_0}(I,a) < 0$, we can prune the game tree beyond action $a$ from iteration $t_0$ until iteration $t_1$ so long as $\sum_{t = 1}^{t_0}v^{\sigma^t}(I,a) + \sum_{t = t_0 + 1}^{t_1} \pi_{-i}^{\sigma^t}(I) U(I,a) \le \sum_{t=1}^{t_1} v^{\sigma^t}(I)$.

\vspace{-0.075in}
\section{\emph{Total RBP}: A New Form of Regret-Based Pruning}
\vspace{-0.075in}
\label{sec:warmrbp}

This section introduces a new form of RBP which we coin \emph{Total RBP}.
When pruning ends and regret must be updated in the pruned branch, Interval RBP calculates a CBR to the average opponent strategy over the skipped iterations, and updates regret in the pruned branch as if that CBR strategy were played on each of the skipped iterations. By contrast, when pruning ends in Total RBP, it calculates a CBR in the pruned branch against the opponent's average strategy over \emph{all} iterations played so far, and sets regret as if that CBR strategy were played on \emph{every} iteration played in the game so far---even those that were played before pruning began.

While using a CBR works correctly in Total RBP, it is also sound to choose a strategy that is \emph{almost} a CBR (formalized later in this section), as long as that strategy does not result in a violation of the CFR bound on the potential function $\Phi(R^T(I))$ of any information set $I$. In practice, this means that the strategy is close to a CBR, and approaches a CBR as $T \rightarrow \infty$. We now present the theory to show that such a near-CBR can be used. However, in practice CFR converges much faster than the theoretical bound, so the potential function is typically far lower than the theoretical bound. Thus, while choosing a near-CBR rather than an exact CBR may allow for slightly longer pruning according to the theory, it may actually result in worse performance. All of the theoretical results presented in this paper, including the improved convergence bound as well as the lower space requirements, still hold if only a CBR is used, and our experiments use a CBR. Nevertheless, clever heuristics for deciding on a near-CBR may lead to even better performance in practice.

We define a strategy $NBR(\sigma_{-i},T)$ as a $T$\emph{-near counterfactual best response} ($T$-near CBR) to $\sigma_{-i}$ if for all $I$ belonging to player $i$
\begin{equation}
\sum_{a \in A(I)} \big(v^{\langle NBR(\sigma_{-i},T), \sigma_{-i} \rangle}(I,a) - v^{\langle NBR(\sigma_{-i},T), \sigma_{-i} \rangle}(I)\big)_+^2 \le \frac{x^T_{I}}{T^2}
\label{eq:nearbest}
\end{equation}
where $x^T_{I}$ can be any value in the range $0 \le x^T_{I} \le y^T_{I}$ and $y^T_{I}$ is the CFR bound on $\Phi(R^T(I))$. If $x^T_{I} = 0$, then a $T$-near CBR is always a CBR. The set of strategies that are $T$-near CBRs to $\sigma_{-i}$ is represented as $\Sigma^{NBR}(\sigma_{-i},T)$. We also define the $T$\emph{-near counterfactual best response value} as $NBV^{\sigma_{-i},T}(I,a) = \min_{\sigma'_i \in \Sigma^{NBR}(\sigma_{-i},T)} v^{\langle \sigma'_i, \sigma_{-i} \rangle}(I,a)$ and $NBV^{\sigma_{-i},T}(I) = \min_{\sigma'_i \in \Sigma^{NBR}(\sigma_{-i},T)} v^{\langle \sigma'_i, \sigma_{-i} \rangle}(I)$.

In Total RBP, an action is pruned only if it would still have negative regret had a $T$-near CBR against the opponent's average strategy been played on every iteration. Specifically, on iteration $T$ of CFR with RM, if
\vspace{-0.09in}
\begin{equation}
T\big(NBV^{\bar{\sigma}^T_{-i},T}(I,a)\big) \le \sum_{t = 1}^T v^{\sigma^t}(I)
\label{eq:rbpcondition}
\end{equation}
then $D(I,a)$ can be pruned for
\vspace{-0.09in}
\begin{equation}
T' = \frac{\sum_{t = 1}^T v^{\sigma^t}(I) - NBV^{\bar{\sigma}^T_{-i},T}(I,a)}{U(I,a) - L(I)}
\label{eq:prunesimple}
\end{equation}
iterations. After those $T'$ iterations are over, we calculate a $T+T'$-near CBR in $D(I,a)$ to the opponent's average strategy and set regret as if that $T+T'$-near CBR had been played on every iteration. Specifically, for each $t \le T + T'$ we set\footnote{In practice, only the sums $\sum_{t=1}^Tv^{\sigma^t}(I)$ and either $\sum_{t=1}^Tv^{\sigma^t}(I,a)$ or $R^T(I,a)$ are stored.} $v^{\sigma^t}(I,a) = NBV^{\bar{\sigma}^{T+T'}_{-i},T+T'}(I,a)$ so that
\begin{equation}
R^{T + T'}(I,a) = \big(T + T'\big)\big(NBV^{\bar{\sigma}^{T+T'}_{-i},T+T'}(I,a)\big) - \sum_{t = 1}^{T + T'} v^{\sigma^t}(I)
\label{eq:pruneregreta}
\end{equation}
and for every information set $I' \in D(I,a)$ we set $v^{\sigma^t}(I',a') = NBV^{\bar{\sigma}^{T+T'}_{-i},T+T'}(I',a')$ and $v^{\sigma^t}(I') = NBV^{\bar{\sigma}^{T+T'}_{-i},T+T'}(I')$ so that
\begin{equation}
R^{T + T'}(I',a') = \big(T + T'\big)\big(NBV^{\bar{\sigma}^{T+T'}_{-i},T+T'}(I',a') - NBV^{\bar{\sigma}^{T+T'}_{-i},T+T'}(I')\big)
\label{eq:pruneregret}
\end{equation}

Theorem~\ref{theorem:totalrbp} proves that if (\ref{eq:rbpcondition}) holds for some action, then the action can be pruned for $T'$ iterations, where $T'$ is defined in (\ref{eq:prunesimple}). The same theorem holds if one replaces the $T$-near counterfactual best response values with counterfactual best response values. The proof for Theorem~\ref{theorem:totalrbp} draws from recent work on warm starting CFR using only an average strategy profile~\cite{Brown16:Strategy-Based}. Essentially, we warm start regrets in the pruned branch using only the average strategy of the opponent and knowledge of $T$.
\begin{theorem}
Assume $T$ iterations of CFR with RM have been played in a two-player zero-sum game and assume $T\big(NBV^{\bar{\sigma}^{T}_{-i},T}(I,a)\big) \le \sum_{t=1}^T v^{\sigma^t}(I)$ where $P(I) = i$. Let $T' = \lfloor \frac{\sum_{t = 1}^T v^{\sigma^t}(I) - T\big(NBV^{\bar{\sigma}^{T}_{-i},T}(I,a)\big)}{U(I,a) - L(I)} \rfloor$. If both players play according to CFR with RM for the next $T'$ iterations in all information sets $I'' \not \in D(I,a)$ except that $\sigma(I,a)$ is set to zero and $\sigma(I)$ is renormalized, then $(T + T')\big(NBV^{\bar{\sigma}^{T+T'}_{-i},T+T'}(I,a)\big) \le \sum_{t=1}^{T+T'} v^{\sigma^t}(I)$. Moreover, if one then sets $v^{\sigma^t}(I,a) = NBV^{\bar{\sigma}^{T+T'}_{-i},T+T'}(I,a)$ for each $t \le T + T'$ and $v^{\sigma^t}(I',a') = NBV^{\bar{\sigma}^{T+T'}_{-i},T+T'}(I',a')$ for each $I' \in D(I,a)$, then after $T''$ additional iterations of CFR with RM, the bound on exploitability of $\bar{\sigma}^{T + T' + T''}$ is no worse than having played $T + T' + T''$ iterations of CFR with RM without RBP.
\label{theorem:totalrbp}
\end{theorem}


In practice, rather than check whether (\ref{eq:rbpcondition}) is met for an action on every iteration, one could only check actions that have very negative regret, and do a check only once every several iterations. This would still be safe and would save some computational cost of the checks, but would lead to less pruning.

As is the case with Interval RBP, the duration of pruning can be increased by giving up knowledge beforehand of exactly how many iterations can be skipped. From (\ref{eq:counterfactuala}) and (\ref{eq:counterfactual}) we see that $r^T(I,a) \le \pi_{-i}^{\sigma^t}(I) \big(U(I,a) - L(I)\big)$. Thus, if $\pi_{-i}^{\sigma^t}(I)$ is very low, then (\ref{eq:rbpcondition}) would continue to hold for more iterations than (\ref{eq:prunesimple}) guarantees. Specifically, we can prune $D(I,a)$ from iteration $t_0$ until iteration $t_1$ as long as
\vspace{-0.1in}
\begin{equation}
t_0 \big(NBV^{\bar{\sigma}^{t_0}_{-i},t_0}(I,a)\big) + \sum_{t = t_0 + 1}^{t_1} \pi_{-i}^{\sigma^t}(I) U(I,a) \le \sum_{t=1}^{t_1} v^{\sigma^t}(I)
\label{eq:totalrbp}
\end{equation}
\vspace{-0.2in}
\subsection{Total RBP Requires Less Space}
\label{sec:rbpspace}
A key advantage of Total RBP is that setting the new regrets according to (\ref{eq:pruneregreta}) and (\ref{eq:pruneregret}) requires no knowledge of what the regrets were before pruning began. Thus, once pruning begins, all the regrets in $D(I,a)$ can be discarded and the space that was allocated to storing the regret can be freed. That space need only be re-allocated once (\ref{eq:totalrbp}) ceases to hold \emph{and} we cannot immediately begin pruning again (that is, (\ref{eq:rbpcondition}) does not hold). Theorem~\ref{theorem:rbpspace} proves that for any information set $I$ and action $a \in A(I)$ that is not part of a best response to a Nash equilibrium, there is an iteration $T_{I,a}$ such that for all $T \ge T_{I,a}$, action $a$ in information set $I$ (and its descendants) can be pruned.\footnote{If CFR converges to a \emph{particular} Nash equilibrium, then this condition could be broadened to any information set $I$ and action $a \in A(I)$ that is not a best response to \emph{that particular} Nash equilibrium. While empirically CFR does appear to always converge to a particular Nash equilibrium, there is no known proof that it always does so.} Thus, once this $T_{I,a}$ is reached, it will never be necessary to allocate space for regret in $D(I,a)$ again.
\begin{theorem}
In a two-player zero-sum game, if for every opponent Nash equilibrium strategy $\sigma_{-P(I)}^*$, $CBV^{\sigma_{-P(I)}^*}(I,a) < CBV^{\sigma_{-P(I)}^*}(I)$, then there exists a $T_{I,a}$ and $\delta_{I,a} > 0$ such that after $T \ge T_{I,a}$ iterations of CFR, $CBV^{\bar{\sigma}^T_{-i}}(I,a) - \frac{\sum_{t = 1}^T v^{\sigma^t}(I)}{T} \le -\delta_{I,a}$.
\label{theorem:rbpspace}
\end{theorem}

While such a constant $T_{I,a}$ exists for any suboptimal action, Total RBP cannot determine whether or when $T_{I,a}$ is reached. So, it is still necessary to check whether (\ref{eq:rbpcondition}) is satisfied whenever (\ref{eq:totalrbp}) no longer holds, and to recalculate how much longer $D(I,a)$ can safely be pruned. This requires the algorithm to periodically calculate a best response (or near-best response) in $D(I,a)$. However, this (near-)best response calculation does not require knowledge of regret in $D(I,a)$, so it is still never necessary to store regret after iteration $T_{I,a}$ is reached.

While it is possible to discard regrets in $D(I,a)$ without penalty once pruning begins, regret is only half the space requirement of CFR. Every information set $I$ also stores a sum of the strategies played $\sum_{t = 1}^T \big(\pi_i^{\sigma^t}(I) \sigma^t(I)\big)$ which is normalized once CFR ends in order to calculate $\bar{\sigma}^T(I)$. Fortunately, if action $a$ in information set $I$ is pruned for long enough, then the stored cumulative strategy in $D(I,a)$ can also be discarded at the cost of a small increase in the distance of the final average strategy from a Nash equilibrium. Specifically, if $\pi_i^{\bar{\sigma}^T}(I,a) \le \frac{C}{\sqrt{T}}$, where $C$ is some constant, then setting $\bar{\sigma}^T(I,a) = 0$ and renormalizing $\bar{\sigma}^T(I)$, and setting $\bar{\sigma}^T(I',a') = 0$ for $I' \in D(I,a)$, can result in at most $\frac{C |I| \Delta}{\sqrt{T}}$ higher exploitability for the whole strategy $\bar{\sigma}^T$. Since CFR only guarantees that $\bar{\sigma}^T$ is a $\frac{2 |\mathcal{I}| \Delta \sqrt{|A|}}{\sqrt{T}}$-Nash equilibrium anyway, $\frac{C |I| \Delta}{\sqrt{T}}$ is only a constant factor of the bound. If an action is pruned from $T'$ to $T$, then $\sum_{t=1}^T \big(\pi_i^{\sigma^t}(I)\sigma^t(I,a)\big) \le \frac{T'}{T}$. Thus, if an action is pruned for long enough, then eventually $\sum_{t=1}^T \big(\pi_i^{\sigma^t}(I)\sigma^t(I,a)\big) \le \frac{C}{\sqrt{T}}$ for any $C$, so $\sum_{t=1}^T \big(\pi_i^{\sigma^t}(I)\sigma^t(I,a)\big)$ could be set to zero (as well as all descendants of $I \cdot a$), while suffering at most a constant factor increase in exploitability. As more iterations are played, this penalty will continue to decrease and eventually be negligible. The constant $C$ can be set by the user: a higher $C$ allows the average strategy to be discarded sooner, while a lower $C$ reduces the potential penalty in exploitability.

We define $\mathcal{I}_{S}$ as the set of information sets that are not guaranteed to be asymptotically pruned by Theorem~\ref{theorem:rbpspace}. Specifically, $I \in \mathcal{I}_{S}$ if $I \not \in D(I',a')$ for some $I'$ and $a' \in A(I')$ such that for every opponent Nash equilibrium strategy $\sigma^*_{-P(I')}$, $CBV^{\sigma_{-P(I')}^*}(I',a') < CBV^{\sigma_{-P(I')}^*}(I')$. Theorem~\ref{theorem:rbpspace} implies the following.
\begin{corollary}
In a two-player zero-sum game with some threshold on the average strategy $\frac{C}{\sqrt{T}}$ for $C > 0$, after a finite number of iterations, CFR with Total RBP requires only $O\big(|\mathcal{I}_S||A|\big)$ space.
\label{corollary:rbpspace}
\end{corollary}

\subsection{Total RBP Has a Better Convergence Bound}
\label{sec:rbpbound}
We now prove that Total RBP in CFR speeds up convergence to an $\epsilon$-Nash equilibrium. Section~\ref{sec:warmrbp} proved that CFR with Total RBP converges in the same number of iterations as CFR alone. In this section, we prove that Total RBP allows each iteration to be traversed more quickly. Specifically, if an action $a \in A(I)$ is not a CBR to a Nash equilibrium, then $D(I,a)$ need only be traversed $O(\ln(T))$ times over $T$ iterations. Intuitively, as both players converge to a Nash equilibrium, actions that are not a counterfactual best response will eventually do worse than actions that are, so those suboptimal actions will accumulate increasing amounts of negative regret. This negative regret allows the action to be safely pruned for increasingly longer periods of time.

Specifically, let $S \subseteq H$ be the set of histories where $h \cdot a \in S$ if $h \in S$ and action $a$ is part of some CBR to some Nash equilibrium. Formally, $S$ contains $\emptyset$ and every history $h \cdot a$ such that $h \in S$ and $CBV^{\sigma^*_{-P(I)}}(I,a) = CBV^{\sigma^*_{-P(I)}}(I)$ for some Nash equilibrium $\sigma^*$.

\begin{theorem}
In a two-player zero-sum game, if both players choose strategies according to CFR with Total RBP, then conducting $T$ iterations requires only $O\big(|S| T + |H|\ln(T)\big)$ nodes to be traversed.
\label{theorem:rbpfast}
\end{theorem}

The definition of $S$ uses properties of the Nash equilibria of the game, and an action $a \in A(I)$ not in $S$ is only guaranteed to be pruned by RBP after some $T_{I,a}$ is reached, which also depends on the Nash equilibria of the game. Since CFR converges to only an $\epsilon$-Nash equilibrium, CFR cannot determine with certainty which nodes are in $S$ or when $T_{I,a}$ is reached. Nevertheless, both $S$ and $T_{I,a}$ are fixed properties of the game.
\section{Experiments}
We compare Total RBP to Interval RBP, to only partial pruning, and to no pruning at all. We also show the amount of space used by Total RBP over the course of the run.  The experiments are conducted on Leduc Hold'em~\cite{Southey05:Bayes} and Leduc-5~\cite{Brown15:Regret-Based}. Leduc Hold'em is a common benchmark in imperfect-information game solving because it is small enough to be solved but still strategically complex. In Leduc Hold'em, there is a deck consisting of six cards: two each of Jack, Queen, and King. There are two rounds. In the first round, each player places an ante of $1$ chip in the pot and receives a single private card. A round of betting then takes place with a two-bet maximum, with Player 1 going first. A public shared card is then dealt face up and another round of betting takes place. Again, Player 1 goes first, and there is a two-bet maximum. If one of the players has a pair with the public card, that players wins. Otherwise, the player with the higher card wins. The bet size in the first round is $2$ chips, and $4$ chips in the second round. Leduc-5 is like Leduc Hold'em but larger: there are 5 bet sizes to choose from. In the first round a player may bet $0.5$, $1$, $2$, $4$, or $8$ chips, while in the second round a player may bet $1$, $2$, $4$, $8$, or $16$ chips.

Results are presented for both CFR and CFR+. Nodes touched is a hardware and implementation-independent proxy for time. Overhead costs are counted in nodes touched. CFR+ is a variant of CFR in which a floor on regret is set at zero and each iteration is weighted linearly in the average strategy (that is, iteration $t$ is weighted by $t$) rather than each iteration being weighted equally. Since Interval RBP can only prune negative-regret actions, Interval RBP modifies the definition of CFR+ so that regret can be negative, but immediately jumps up to zero as soon as regret increases. Total RBP does not require this modification. Both forms of RBP modify the behavior of CFR+ because without pruning, CFR+ would put positive probability on an action as soon as its regret increases, while RBP waits until pruning is over. This is not, by itself, a problem. However, CFR+'s linear weighting of the average strategy is only guaranteed to converge to a Nash equilibrium if pruning does not occur. While pruning does well empirically with CFR+, the convergence is noisy. This noise can be reduced by using the lowest-exploitability average strategy profile found so far. We do this in the experiments.

Figure~\ref{fig:rbpspace} shows the reduction in space needed to store the average strategy and regrets for Total RBP---for various values of the constant threshold $C$, where an action's probability is set to zero if it is reached with probability less than $\frac{C}{\sqrt{T}}$ in the average strategy, as we explained in Section~\ref{sec:rbpspace}. In both games, a threshold between 0.01 and 0.1 performed well in both space and number of iterations, with the lower thresholds converging somewhat faster and the higher thresholds reducing space faster. We also tested thresholds below 0.01, but the speed of convergence was essentially the same as when using 0.01. In Leduc, all variants resulted in a quick drop-off in space to about half the initial amount. In Leduc-5, a threshold of 0.1 resulted in a factor of 10 reduction in space for CFR+, and a factor of 7 reduction for CFR. In the case of CFR, this space reduction factor appears to continue to increase.
\begin{figure}[!h]
	\centering
	\subfloat[Leduc Hold'em]{{\includegraphics[width=6.97cm]{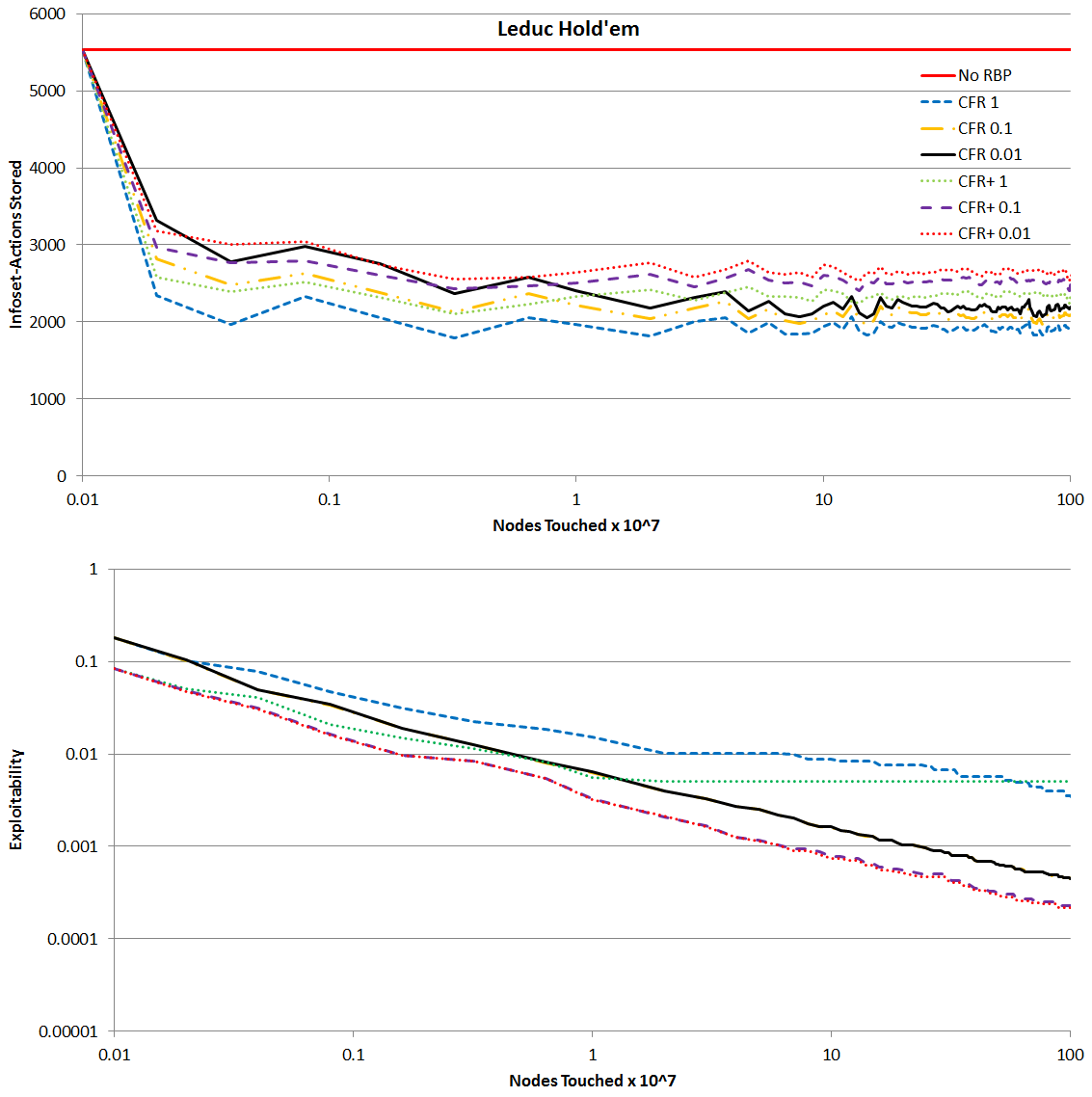}}}
	\subfloat[Leduc-5 Hold'em]{{\includegraphics[width=6.97cm]{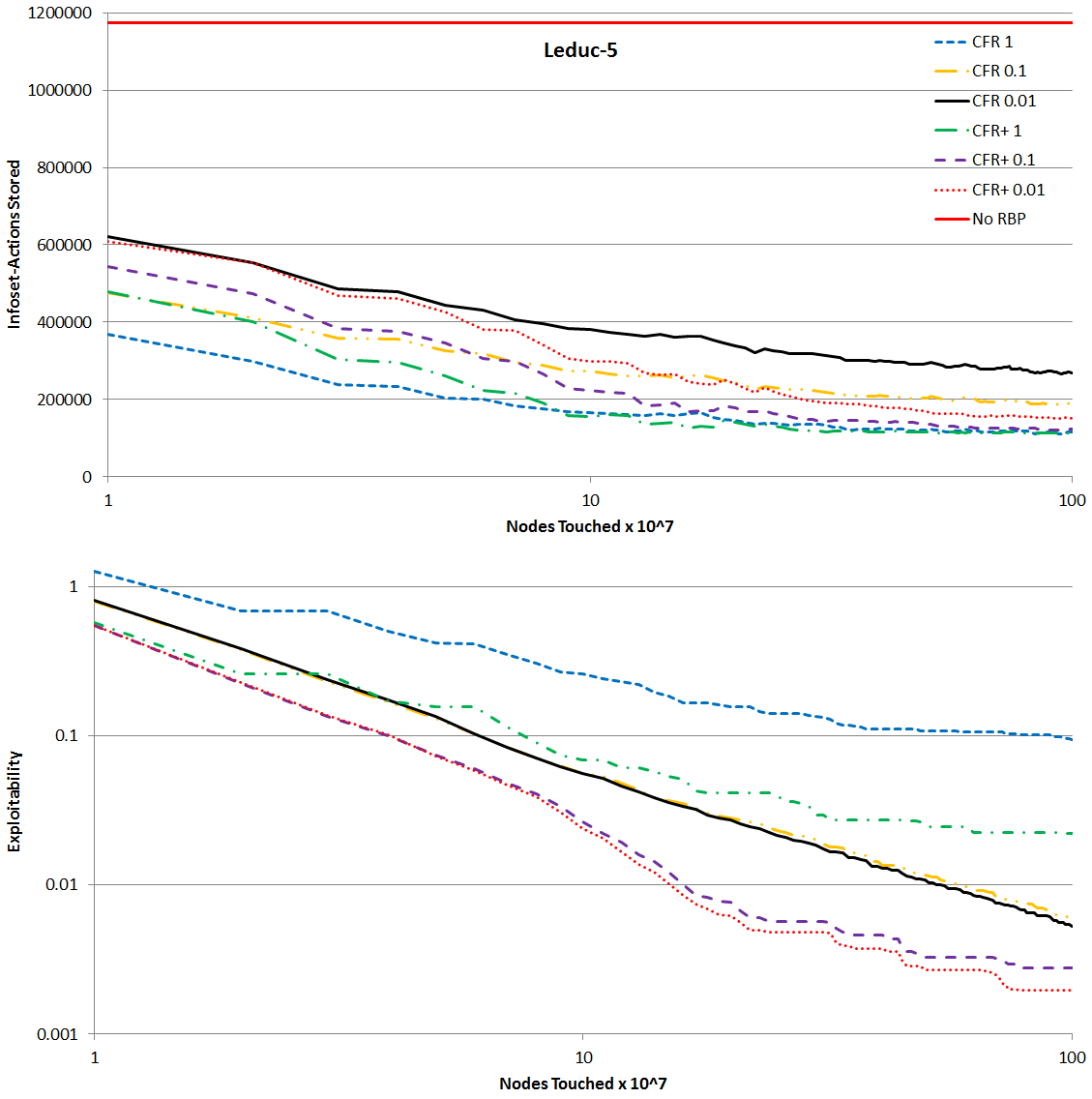}}}
	\caption{Convergence and space required for CFR and CFR+ with Total RBP.}
	\label{fig:rbpspace}
\end{figure}

Figure~\ref{fig:rbpconverge} compares the convergence rates of Total RBP, Interval RBP, and only partial pruning for both CFR and CFR+. In Leduc, Total RBP and Interval RBP perform comparably when added to CFR. When added to CFR+, Interval RBP does significantly better. In Leduc-5, which is a far larger game, Total RBP outperforms Interval RBP by a factor of 2 when added to CFR. When added to CFR+, Total RBP initially does far better but its performance is eventually surpassed by Interval RBP. This may be due to the noisy performance of CFR+ with RBP.
\begin{figure}[!h]
	\centering
	\subfloat[Leduc Hold'em]{{\includegraphics[width=6.97cm]{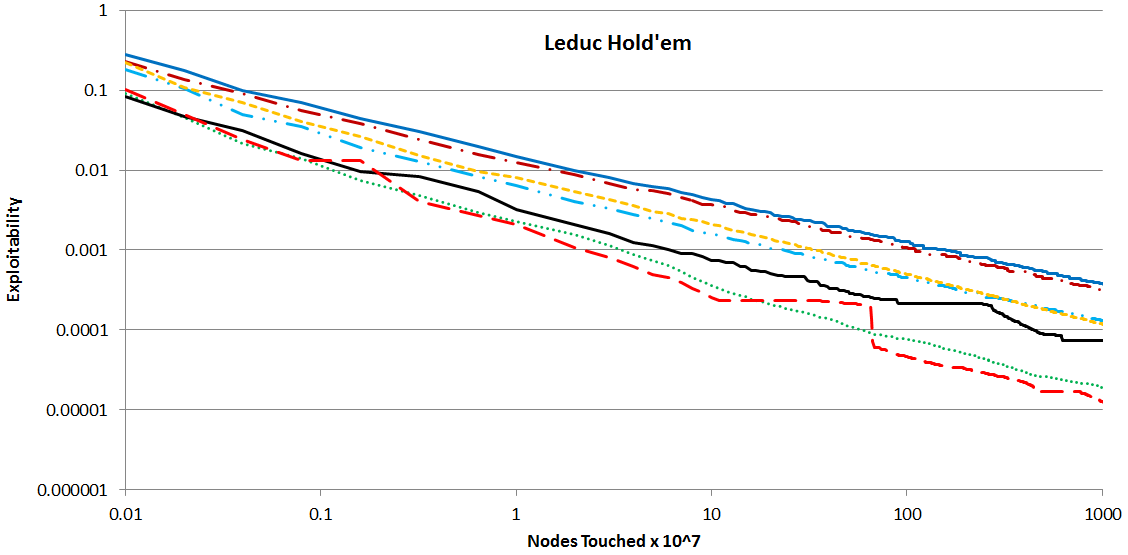}}}
	\subfloat[Leduc-5 Hold'em]{{\includegraphics[width=6.97cm]{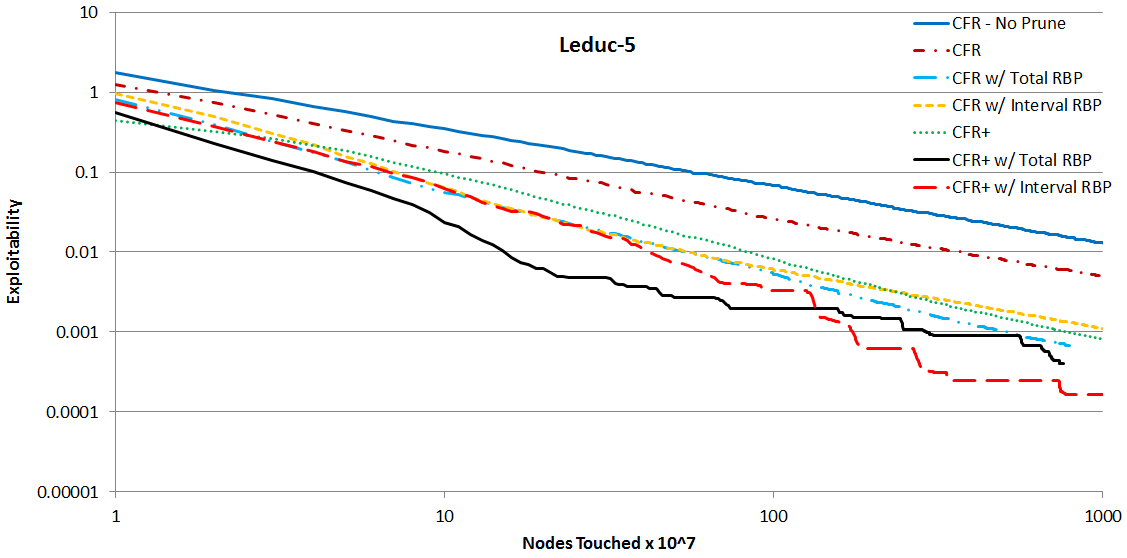}}}
	\caption{Convergence for CFR and CFR+ with only partial pruning, with Interval RBP, and with Total RBP. ``CFR - No Prune'' is CFR without any pruning.}
	\label{fig:rbpconverge}
\end{figure}

\vspace{-0.075in}
\section{Conclusions}
We introduced Total RBP, a new form of regret-based pruning that provably reduces both the space needed to solve an imperfect-information game and the time needed to reach an $\epsilon$-Nash equilibrium. This addresses both of the major bottlenecks in solving large imperfect-information games. Experimentally, Total RBP reduced the space needed to solve a game by an order of magnitude, with the reduction factor increasing with game size.

\newpage

\bibliographystyle{plain}

\begin{thebibliography}{10}
	
	\bibitem{Bowling15:Heads-up}
	Michael Bowling, Neil Burch, Michael Johanson, and Oskari Tammelin.
	\newblock Heads-up limit hold'em poker is solved.
	\newblock {\em Science}, 347(6218):145--149, January 2015.
	
	\bibitem{Brown15:Regret-Based}
	Noam Brown and Tuomas Sandholm.
	\newblock Regret-based pruning in extensive-form games.
	\newblock In {\em Proceedings of the Annual Conference on Neural Information
		Processing Systems (NIPS)}, 2015.
	
	\bibitem{Brown15:Simultaneous}
	Noam Brown and Tuomas Sandholm.
	\newblock Simultaneous abstraction and equilibrium finding in games.
	\newblock In {\em Proceedings of the International Joint Conference on
		Artificial Intelligence (IJCAI)}, 2015.
	
	\bibitem{Brown16:Strategy-Based}
	Noam Brown and Tuomas Sandholm.
	\newblock Strategy-based warm starting for regret minimization in games.
	\newblock In {\em AAAI Conference on Artificial Intelligence (AAAI)}, 2016.
	
	\bibitem{Gilpin07:Lossless_long}
	Andrew Gilpin and Tuomas Sandholm.
	\newblock Lossless abstraction of imperfect information games.
	\newblock {\em Journal of the ACM}, 54(5), 2007.
	\newblock Early version `Finding equilibria in large sequential games of
	imperfect information' appeared in the Proceedings of the ACM Conference on
	Electronic Commerce (EC), pages 160--169, 2006.
	
	\bibitem{Hart00:Simple}
	Sergiu Hart and Andreu Mas-Colell.
	\newblock A simple adaptive procedure leading to correlated equilibrium.
	\newblock {\em Econometrica}, 68:1127--1150, 2000.
	
	\bibitem{Hoda10:Smoothing}
	Samid Hoda, Andrew Gilpin, Javier Pe{\~n}a, and Tuomas Sandholm.
	\newblock Smoothing techniques for computing {N}ash equilibria of sequential
	games.
	\newblock {\em Mathematics of Operations Research}, 35(2):494--512, 2010.
	\newblock Conference version appeared in WINE-07.
	
	\bibitem{Kroer15:Faster}
	Christian Kroer, Kevin Waugh, Fatma K{\i}l{\i}n\c{c}-Karzan, and Tuomas
	Sandholm.
	\newblock Faster first-order methods for extensive-form game solving.
	\newblock In {\em Proceedings of the ACM Conference on Economics and
		Computation (EC)}, 2015.
	
	\bibitem{Lanctot09:Monte}
	Marc Lanctot, Kevin Waugh, Martin Zinkevich, and Michael Bowling.
	\newblock {M}onte {C}arlo sampling for regret minimization in extensive games.
	\newblock In {\em Proceedings of the Annual Conference on Neural Information
		Processing Systems (NIPS)}, pages 1078--1086, 2009.
	
	\bibitem{Moravcik16:Refining}
	Matej Moravcik, Martin Schmid, Karel Ha, Milan Hladik, and Stephen~J
	Gaukrodger.
	\newblock Refining subgames in large imperfect information games.
	\newblock In {\em Thirtieth AAAI Conference on Artificial Intelligence}, 2016.
	
	\bibitem{Nesterov05:Excessive}
	Yurii Nesterov.
	\newblock Excessive gap technique in nonsmooth convex minimization.
	\newblock {\em SIAM Journal of Optimization}, 16(1):235--249, 2005.
	
	\bibitem{Pays14:Interior}
	Fran{\c{c}}ois Pays.
	\newblock An interior point approach to large games of incomplete information.
	\newblock In {\em AAAI Computer Poker Workshop}, 2014.
	
	\bibitem{Southey05:Bayes}
	Finnegan Southey, Michael Bowling, Bryce Larson, Carmelo Piccione, Neil Burch,
	Darse Billings, and Chris Rayner.
	\newblock Bayes' bluff: Opponent modelling in poker.
	\newblock In {\em Proceedings of the 21st Annual Conference on Uncertainty in
		Artificial Intelligence (UAI)}, pages 550--558, July 2005.
	
	\bibitem{Tammelin15:Solving}
	Oskari Tammelin, Neil Burch, Michael Johanson, and Michael Bowling.
	\newblock Solving heads-up limit texas hold'em.
	\newblock In {\em Proceedings of the 24th International Joint Conference on
		Artificial Intelligence (IJCAI)}, 2015.
	
	\bibitem{Waugh09:Abstraction}
	Kevin Waugh, David Schnizlein, Michael Bowling, and Duane Szafron.
	\newblock Abstraction pathologies in extensive games.
	\newblock In {\em International Conference on Autonomous Agents and Multi-Agent
		Systems (AAMAS)}, 2009.
	
	\bibitem{Zinkevich07:Regret}
	Martin Zinkevich, Michael Bowling, Michael Johanson, and Carmelo Piccione.
	\newblock Regret minimization in games with incomplete information.
	\newblock In {\em Proceedings of the Annual Conference on Neural Information
		Processing Systems (NIPS)}, 2007.
	
\end{thebibliography}

\clearpage

\newpage

\appendix
\noindent {\LARGE \bf Appendix}

In the appendices we present the proofs, and additional lemmas that are used in the proofs.

\section{Lemma~\ref{lemma:rbpwarm}}

Lemma~\ref{lemma:rbpwarm} proves that if (\ref{eq:rbpcondition}) is satisfied for some action $a \in A(I)$ on iteration $T$, then the value of action $a$ and all its descendants on every iteration played so far can be set to the $T$-near counterfactual best response value. The same lemma holds if one replaces the $T$-near counterfactual best response values with exact counterfactual best response values. The proof for Lemma~\ref{lemma:rbpwarm} draws from recent work on warm starting CFR using only an average strategy profile~\cite{Brown16:Strategy-Based}.

\begin{lemma}
	Assume $T$ iterations of CFR with RM have been played in a two-player zero-sum game.
	If $T\big(NBV^{\bar{\sigma}^{T}_{-i},T}(I,a)\big) \le \sum_{t=1}^T v^{\sigma^t}(I)$ and one sets $v^{\sigma^t}(I,a) = NBV^{\bar{\sigma}^{T}_{-i},T}(I,a)$ for each $t \le T$ and for each $I' \in D(I,a)$ sets $v^{\sigma^t}(I',a') = NBV^{\bar{\sigma}^{T}_{-i},T}(I',a')$ and $v^{\sigma^t}(I') = NBV^{\bar{\sigma}^{T}_{-i},T}(I')$ then after $T'$ additional iterations of CFR with RM, the bound on exploitability of $\bar{\sigma}^{T + T'}$ is no worse than having played $T + T'$ iterations of CFR with RM unaltered.
	\label{lemma:rbpwarm}
\end{lemma}
\begin{proof}
	The proof builds upon Theorem~2 in \cite{Brown16:Strategy-Based}. Assume $T\big(NBV^{\bar{\sigma}^{T}_{-i},T}(I,a)\big) \le \sum_{t=1}^T v^{\sigma^t}(I)$. We wish to warm start to $T$ iterations. For each $I' \in D(I,a)$ set $v^{\sigma^t}(I',a') = NBV^{\bar{\sigma}^{T}_{-i},T}(I',a')$ and $v^{\sigma^t}(I') = NBV^{\bar{\sigma}^{T}_{-i},T}(I')$ and set $v^{\sigma^t}(I,a) = NBV^{\bar{\sigma}^{T}_{-i},T}(I,a)$ for all $t \le T$. For every other action, leave regret unchanged. For each $I' \in D(I,a)$ we know by construction that $\Phi(R^T(I'))$ is within the CFR bound $y_{I'}^T$ after changing regret. By assumption $T\big(NBV^{\bar{\sigma}^{T}_{-i},T}(I,a)\big) \le \sum_{t=1}^T v^{\sigma^t}(I)$, so $R^T(I,a) \le 0$ and therefore $\Phi(R^T(I))$ is unchanged. Finally, since the $T$ iterations were played according to CFR with RM and regret is unchanged for every other information set $I''$, so
	the conditions for Theorem~2 in \cite{Brown16:Strategy-Based} hold for every information set, and therefore we can warm start to $T$ iterations of CFR with RM with no penalty to the convergence bound.
\end{proof}

\section{Proof of Theorem~\ref{theorem:totalrbp}}
\begin{proof}
	From Lemma~\ref{lemma:rbpwarm} we can immediately set regret for $a \in A(I)$ to $v^{\sigma^t}(I,a) = NBV^{\bar{\sigma}^{T}_{-i},T}(I,a)$. By construction of $T'$, $R^t(I,a)$ is guaranteed to be nonpositive for $T \le t \le T + T'$ and therefore $\sigma^t(I,a) = 0$. Thus, $\bar{\sigma}^{T + T'}_i(I')$ for $I' \in D(I,a)$ is identical regardless of what is played in $D(I,a)$ during $T \le t \le T + T'$.
	
	Since $(T + T')\big(NBV^{\bar{\sigma}^{T+T'}_{-i},T+T'}(I,a)\big) \le T\big(NBV^{\bar{\sigma}^T_{-i},T}(I,a)\big) + T' \big(U(I,a)\big)$ and $\sum_{t=1}^{T+T'} v^{\sigma^t}(I) \ge \sum_{t=1}^{T} v^{\sigma^t}(I) + T' \big(L(I)\big)$, so by the definition of $T'$, $(T + T')\big(NBV^{\bar{\sigma}^{T+T'}_{-i},T+T'}(I,a)\big) \le \sum_{t=1}^{T+T'} v^{\sigma^t}(I)$. So if regrets in $D(I,a)$ and $R^{T + T'}(I,a)$ are set according to Lemma~\ref{lemma:rbpwarm}, then after $T''$ additional iterations of CFR with RM, the bound on exploitability of $\bar{\sigma}^{T + T' + T''}$ is no worse than having played $T + T' + T''$ iterations of CFR with RM from scratch.
\end{proof}

\section{Proof of Theorem~\ref{theorem:rbpspace}}
\begin{proof}
	Consider an information set $I$ and action $a \in A(I)$ where for every opponent Nash equilibrium strategy $\sigma_{-P(I)}^*$, $CBV^{\sigma_{-P(I)}^*}(I,a) < CBV^{\sigma_{-P(I)}^*}(I)$. Let $i = P(I)$. Let $\delta = \min_{\sigma_{-i} \in \Sigma^*} \big(CBV^{\sigma_{-i}}(I) - CBV^{\sigma_{-i}}(I,a)\big)$ where $\Sigma^*$ is the set of Nash equilibria.
	Let $\sigma'_{-i} = \argmax_{\sigma_{-i} \in \Sigma_{-i} \mid CBV^{\sigma_{-i}}(I) - CBV^{\sigma_{-i}}(I,a) \le \frac{3 \delta}{4}} u_{-i} (\sigma_{-i}, BR(\sigma_{-i}))$ Since $\sigma_{-i}'$ is not a Nash equilibrium strategy and CFR converges to a Nash equilibrium strategy for both players, so
	there exists a $T_{\delta}$ such that for all $T \ge T_{\delta}$, $CBV^{\bar{\sigma}^T_{-i}}(I) - CBV^{\bar{\sigma}_{-i}^T}(I,a) > \frac{3 \delta}{4}$. Let $T'_{I,a} = \frac{4|\mathcal{I}|^2\Delta^2|A|}{\delta^2}$. For $T \ge T'_{I,a}$ since $R_i^T \le \sum_{I \in \mathcal{I}_i} R^T(I)$, so $CBV^{\bar{\sigma}_{-i}^T}(I) - \sum_{t = 1}^T v^{\sigma^t}(I) \le \frac{\delta}{2}$. Let $T_{I,a} = \max(T'_{I,a}, T_{\delta})$ and $\delta_{I,a} = \frac{\delta}{4}$. Then for $T \ge T_{I,a}$, $CBV^{\bar{\sigma}^T_{-i}}(I,a) - \frac{\sum_{t = 1}^T v^{\sigma^t}(I)}{T} \le -\delta_{I,a}$.
\end{proof}

\section{Proof of Corollary~\ref{corollary:rbpspace}}
\begin{proof}
Let $I \not \in \mathcal{I}_S$. Then $I \in D(I',a')$ for some $I'$ and $a' \in A(I')$ such that for every opponent Nash equilibrium strategy $\sigma^*_{-P(I')}$, $CBV^{\sigma_{-P(I')}^*}(I',a') < CBV^{\sigma_{-P(I')}^*}(I')$. Applying Theorem~\ref{theorem:rbpspace}, this means there exists a $T_{I',a'}$ and $\delta_{I',a'} > 0$ such that for $T \ge T_{I',a'}$, $CBV^{\bar{\sigma}^T_{-i}}(I',a') - \frac{\sum_{t = 1}^T v^{\sigma^t}(I')}{T} \le -\delta_{I',a'}$. So (\ref{eq:rbpcondition}) always applies for $T \ge T_{I',a'}$ for $I'$ and $a'$ and $I$ will always be pruned. Since (\ref{eq:pruneregret}) does not require knowledge of regret, it need not be stored for $I$.

Since $D(I',a')$ will always be pruned for $T \ge T_{I',a'}$, so for any $T \ge \frac{(T_{I',a'})^2}{C^2}$ iterations for some constant $C > 0$, $\pi_i^{\bar{\sigma}^T}(I) \le \frac{C}{\sqrt{T}}$, which satisfies the threshold of the average strategy. Thus, the average strategy in $D(I,a)$ can be discarded.
\end{proof}

\section{Lemma~\ref{lemma:negative}}
\begin{lemma}
	If for all $T \ge T'$ iterations of CFR with RBP, $T \big(CBV^{\bar{\sigma}^T}(I,a)\big) - \sum_{t = 1}^T v^{\sigma^t}(I) \le -xT$ for some $x > 0$, then any history $h'$ such that $h \cdot a \sqsubseteq h'$ for some $h \in I$ need only be traversed at most $O\big(\ln(T)\big)$ times.
	\label{lemma:negative}
\end{lemma}
\begin{proof}
Let $a \in A(I)$ be an action such that for all $T \ge T'$, $T\big(CBV^{\bar{\sigma}^T}(I,a)\big) - \sum_{t = 1}^T v^{\sigma^t}(I) \le -xT$ for some $x > 0$. $NBV^{\bar{\sigma}^T_{-i},T}(I,a) \le CBV^{\bar{\sigma}^T_{-i}}$, so from Theorem~\ref{theorem:totalrbp}, $D(I,a)$ can be pruned for $m \ge \lfloor \frac{xT}{U(I,a) - L(I)} \rfloor$ iterations on iteration $T$. Thus, over iterations $T \le t \le T + m$, only a constant number of traversals must be done. So each iteration requires only $\frac{C}{m}$ work when amortized, where $C$ is a constant. Since $x$, $U(I,a)$, and $L(I)$ are constants, so on each iteration $t \ge T'$, only an average of $\frac{C}{t}$ traversals of $D(I,a)$ is required. Summing over all $t \le T$ for $T \ge T'$, and recognizing that $T'$ is a constant, we get that action $a$ is only taken $O\big(\ln(T)\big)$ over $T$ iterations. Thus, any history $h'$ such that $h \cdot a \sqsubseteq h'$ for some $h \in I$ need only be traversed at most $O\big(\ln(T)\big)$ times.
\end{proof}

\section{Proof of Theorem~\ref{theorem:rbpfast}}
\begin{proof}
	Consider an $h^* \not \in S$. Then there exists some $h \cdot a \sqsubseteq h^*$ such that $h \in S$ but $h \cdot a \not \in S$. Let $I = I(h)$ and $i = P(I)$. Since $h \cdot a \not \in S$ but $h \in S$, so for every Nash equilibrium $\sigma^*$, $CBV^{\sigma^*}(I,a) < CBV^{\sigma^*}(I)$. From Theorem~\ref{theorem:rbpspace}, there exists a $T_{I,a}$ and $\delta_{I,a} > 0$ such that after $T \ge T_{I,a}$ iterations of CFR, $CBV^{\bar{\sigma}^T_{-i}}(I,a) - \frac{\sum_{t = 1}^T v^{\sigma^t(I)}}{T} \le -\delta_{I,a}$. Thus from Lemma~\ref{lemma:negative}, $h^*$ need only be traversed at most $O\big(\ln(T)\big)$ times.
\end{proof}
\end{document}